\documentclass{article}
\usepackage{amsmath}
\usepackage{amsfonts}
\usepackage{amssymb}
\usepackage[dvips]{color}
\usepackage{graphicx}%
\newtheorem{Theorem}{Theorem}[section]
\newtheorem{pro}[Theorem]{Proposition}
\newtheorem{thm}[Theorem]{Theorem}
\newtheorem{dfn}[Theorem]{Definition}

\newtheorem{lem}[Theorem]{Lemma}

\newenvironment{proof}[1][Proof]{\noindent\textbf{#1.} }{\ \rule{0.5em}{0.5em}}
\newcommand{\bpartial}{\mathop{\partial\kern -4pt\raisebox{.8pt}{$|$}}}
\newcommand{\bra}{\mathopen{[\kern-1.6pt[}}
\newcommand{\ket}{\mathclose{]\kern-1.5pt]}}
\newcommand{\bbra}{\mathopen{[\kern-2.2pt[\kern-2.3pt[}}
\newcommand{\bket}{\mathclose{]\kern-2.1pt]\kern-2.3pt]}}

\makeindex
\makeatletter
\newcommand*{\rom}[1]{\expandafter\@slowromancap\romannumeral #1@}
\makeatother
           
\begin{document}

\title {\large{\bf   Jacobi-Lie Hamiltonian systems  
on
 real low-dimensional 
Jacobi-Lie  groups  and  their
Lie symmetries 
 }}
\vspace{3mm}
\author { \small{ \bf H. Amirzadeh-Fard$^1$ }\hspace{-2mm}{ \footnote{ e-mail: h.amirzadehfard@azaruniv.ac.ir}}, 
 \small{ \bf  Gh. Haghighatdoost$^2$}\hspace{-1mm}{ \footnote{ e-mail: gorbanali@azaruniv.ac.ir}},
  \small{ \bf A.
Rezaei-Aghdam$^3$ }\hspace{-1mm}{\footnote{ e-mail: rezaei-a@azaruniv.ac.ir }}\\
{\small{$^{1,2}$\em Department
of Mathematics, Azarbaijan Shahid Madani University, 53714-161, Tabriz, Iran}}\\
{\small{$^{3}$\em Department of Physics, Azarbaijan Shahid Madani
University, 53714-161, Tabriz, Iran}}\\
  }
 \maketitle
\begin{abstract}

We study Jacobi-Lie   Hamiltonian systems  admitting  Vessiot-Guldberg Lie algebras of Hamiltonian vector
fields related to  Jacobi structures on
 real low-dimensional 
Jacobi-Lie  groups. Also,
 we find some examples
  of  Jacobi-Lie   Hamiltonian systems  on
real two- and three- dimensional Jacobi-Lie groups.
Finally,  we present 
Lie symmetries of Jacobi-Lie  Hamiltonian systems on some three-dimensional real Jacobi-Lie groups.
\end{abstract}

\smallskip

{\bf keywords:}{\;Jacobi-Lie  group, Jacobi manifold, Lie system,  Jacobi-Lie Hamiltonian system, Lie symmetry.

\section {\large {\bf Introduction}}
A Lie system is a non-independent system of  first-order ODEs that possesses  a superposition rule. In other words,
  a Lie system amounts to a  non-autonomous vector field that  takes values in a finite-dimensional real Lie algebra of vector fields,
 referred to as Vessiot-Guldberg Lie algebra 
 (VG Lie algebra)
of the system,  with respect to a geometric
structure.

At the end of the 19th century, the study
 of Lie systems was
carried out by Sophus Lie,
  who  pioneered the study of systems of
ODEs
 admitting superposition rules \textcolor{red}{[\ref{ref8}]}.
 Then about a century, the study of this problem  was silent.
 After the work of
 Winternitz \textcolor{red}{[\ref{ref52}]},
many authors have recently  
 investigated  this problem \textcolor{red}{[\ref{ref53}-\ref{ref212}]}.
Some results have  been obtained for Lie systems admitting
a VG Lie algebra of Hamiltonian vector fields relative to  
symplectic and Poisson structures  
\textcolor{red}{  [\ref{ref7}]}.

 A particular class of Lie systems on Poisson manifolds, the so-called Lie-Hamilton systems, that admit a VG Lie algebra of Hamiltonian vector fields with
respect to a Poisson structure was studied in \textcolor{red}{[\ref{ref7}]}.
Recently Lie systems possessing a VG Lie algebras of Hamiltonian
vector fields with respect  to Jacobi structures \textcolor{red}{[\ref{ref3}, \ref{ref4}, \ref{ref1}]} 
 were 
 referred to as the Jacobi-Lie systems are studied and exactly introduced in \textcolor{red}{[\ref{ref9}]}.

It is well known that
the symplectic
  manifold  is a particular case of the
  Poisson manifold 
so that the  Poisson bracket
 is not necessarily assumed to be non-degenerate.
In order that the Jacobi bracket is not necessarily assumed to be derivation,
 the
Jacobi manifold
  is a  generalization of
  the  Poisson manifold
  \textcolor{red}{[\ref{ref3}, \ref{ref4}]}. 
 To wrap up the discussion,
 Lie-Hamiltonian systems
  are  a  generalization of  Hamiltonian systems
and
 a particular case of
 the   Jacobi-Lie Hamiltonian systems.

In this work, we  study Lie systems with VG Lie algebras
of Hamiltonian vector fields with respect to Jacobi-Lie groups
 \textcolor{red}{[\ref{ref1}]},
 especially
 on real two and three-dimensional Jacobi-Lie
 groups   \textcolor{red}{[\ref{ref213}, \ref{ref2}]}. 

The outline of the  paper  is as follows: 
In section 2, we recall several definitions
and results in  Lie systems  and Jacobi structures on  Jacobi-Lie
 groups  \textcolor{red}{[\ref{ref1}]}
and  Jacobi-Lie  Hamiltonian system on a Jacobi manifold  \textcolor{red}{[\ref{ref9}]}.
  In Section 3 we exemplify results of Sections 2 on real two- and three-dimensional 
Jacobi-Lie  groups \textcolor{red}{[\ref{ref213}, \ref{ref2}]}. 

Finally,  in Section 4 we study 
Lie symmetry of Jacobi-Lie  Hamiltonian system on some three-dimensional real Jacobi-Lie groups.
 \section{ Definitions and Notations on Lie and Jacobi-Lie Hamiltonian system}
For simplicity, 
all functions and geometric
structures throughout this paper
are assumed  to be real, smooth, and globally defined. 
In order to highlight the main aspects of our results, let
 us  omit  minor technical problems. 
 Here,
for self-containment  of the paper, we review some  basic
concepts of Lie,  Lie-Hamiltonian \textcolor{red}{[\ref{ref7}, \ref{ref5}]} and 
Jacobi-Lie Hamiltonian systems \textcolor{red}{ [\ref{ref9}]}. 
\subsection{ Lie systems and Lie-Hamiltonian systems}

Let
$\mathfrak{a}$ and $\mathfrak{b}$
be two vector  subspaces of Lie algebra $\mathfrak{g},$
and let $[\mathfrak{a} ,\mathfrak{b}]$ denote the  vector
space spanned by the Lie brackets between elements of $\mathfrak{a}$ and $\mathfrak{b},$
 respectively.
 We define $Lie(\mathfrak{a},\mathfrak{g} , [., .])$ to be the smallest Lie subalgebra of
$(\mathfrak{g} , [., .])$  containing $\mathfrak{a}$
and represent  it  with
 $Lie(\mathfrak{a})$.
\begin{dfn}\label{df1}
A t-dependent vector field on a manifold  $M$ is a map 
\begin{equation}\label{2.1}
X:  \mathbb{R} \times M\longrightarrow TM, \quad (t, x)\mapsto X(t, x),
\end{equation}
satisfying $ \tau_M \tiny{o} X = \pi_{2},$
 where  
 $ \pi_{2} $ and  $\tau_M$
   are the projections from $\mathbb{R} \times M $ and $TM $ onto $M,$  respectively. 
\end{dfn}

Using this definiition, we can identify every  $t$-dependent vector field
with   a family $\lbrace X_t \rbrace _{t\in \mathbb{R}} $ of vector fields
$
X _{ t}:   M\longrightarrow TM, \quad  x\mapsto X_{t} (x)=X(t, x),
$
and vice versa.

\begin{dfn}\label{}
  The minimal Lie algebra of $ X$ on a manifold $M$ is the smallest
real Lie algebra, $\mathfrak{g}^X,$ containing the  vector fields $\lbrace X_{t}\rbrace _{t\in \mathbb{R}}.$ In other words,
$\mathfrak{g}^X=Lie (\lbrace X_{t}\rbrace _{t\in \mathbb{R}} )$.
\end{dfn}

\begin{dfn}\label{}
An integral curve of $X$ is an 
 integral curve
$\alpha:  \mathbb{R}\longrightarrow \mathbb{R} \times M, \quad  t\mapsto (t, x(t)),$
of the suspension of $X$\textcolor{red} { [\ref{ref6}]}, i.e.  for the vector field
\begin{equation}\label{2.1}
\tilde{X}:  \mathbb{R} \times M\longrightarrow T(\mathbb{R} \times M)\simeq T\mathbb{R}\oplus TM, \qquad (t, x)\mapsto\dfrac{\partial}{\partial t}+X(t, x),
\end{equation}
we have
$\dfrac{dx(t)}{dt}=X(t,x).$
\end{dfn}

\begin{dfn}
A Lie system is a system
$X$
on a manifold  $M$
whose
$\mathfrak{g}^X$
is finite-dimensional \textcolor{red}{[\ref{ref7}]}.
\end{dfn}

\begin{dfn}\label{2.4}
A superposition rule
depending on n particular solutions
 for a system $X$ on a manifold $M$ is a function $\Psi : M^{n} \times M \longrightarrow M,$
 $(x_{(1)}, . . . , x_{(n)}; k)\longmapsto  x,$ 
 such  that the general solution $x(t)$ of $X$ can be
written as
\begin{equation*}
x(t) = \Psi(x_{(1)}(t), . . . , x_{(n)}(t); k), 
\end{equation*}
where $x_{(1)}(t), . . . , x_{(n)}(t) $ is any generic
collection
 of particular solutions to $X$ and $ k =
(k_1, . . . , k_m)$ 
is a point of M to
be related to initial conditions
 of $X$ \textcolor{red}{[\ref{ref8}]}.
\end{dfn}
\begin{thm}\label{th2.2}
{\bf (The Lie-Scheffers Theorem) }
 A system $X$ on $M$ admits a superposition
rule if and only if
 it can
be written in the form
\begin{equation}\label{2.5}
X(t, x)=\sum^r_{i=1} b_i(t)X_i( x)
\end{equation}
for a set $b_1(t), . . . , b_r(t)$ of t-dependent functions and a family of vector fields $X_1, . . . , X_r $ on
M spanning an r-dimensional real Lie algebra: a VG Lie algebra of X.
 That is to say,
 a system X on $M$ possesses a superposition rule if
and only if it is a Lie system
\textcolor{red}{ [\ref{ref8},\ref{ref11111}]}.
\end{thm}
\begin{dfn}
A  manifold M endowed with bivector structure $P \in\Gamma(\large{\bigwedge}^2 TM)$
satisfying that
$[P, P]=0$
is called a Poisson manifold
$(M,P),$
where  
 [., .]
is the Schouten-Nijenhius bracket
 \footnote{$[X_1\wedge...\wedge X_p, Y_1\wedge...\wedge Y_q]=
\sum\limits_{{\rm{i = 1}}}^{\rm{p}} \sum\limits_{{\rm{j = 1}}}^{\rm{q}}(-1)^{i+j}[X_i,X_j] \wedge X_1\wedge...\wedge\hat{X_i}\wedge...\wedge X_p\wedge Y_1\wedge...\wedge\hat{Y_j}\wedge...\wedge Y_q$}
 and
$\Gamma(\bigwedge^2 TM)$
is the space of sections of 
$\bigwedge^2 TM$\textcolor{red}{ [\ref{11}]}.
\end{dfn}

The  bivector $P$ has a bundle morphism
$P^\#:TM^{*}\longrightarrow TM$
is defined by
$\beta (P^{\#}\alpha)=P(\alpha, \beta)\quad  \forall  \alpha, \beta \in TM^{*}.$\\
\begin{dfn}
A  vector field X on M  with the
bivector  $P $
is said to be  a Hamiltonian vector field 
if
 it can be written as
 $X={P}^\# (df)$
where f is a  function on M, called the Hamiltonian;
conversely, every function f is called the Hamiltonian  function of a unique
Hamiltonian vector field $X_f$.
\end{dfn}

\begin{dfn} 
A Lie-Hamiltonian system  on a Poisson manifold $M$ is a Lie system X whose 
$\mathfrak{g}^X$
 consists of
Hamiltonian vector fields relative to a Poisson bivector 
$P$
\textcolor{red}{[\ref{ref7}]}. 
\end{dfn}


\subsection{  Jacobi and  Jacobi-Lie  structures  }
The study of  the Jacobi manifolds was made  by Lichnerowicz and Kirillov \textcolor{red}{[\ref{ref3}, \ref{ref4}]}.
 
\begin{dfn}\label{df1}
A Jacobi manifold is a triple $(M, \Lambda, E)$, where $\Lambda$ is a $2$-vector field  and $E$
is a vector field on $M$  (called the Reeb vector
field)  such that
\begin{equation}
[\Lambda, \Lambda]=2E\wedge \Lambda,\qquad L_{E}\Lambda=[E, \Lambda]=0,
\end{equation}
where 
the bracket is that of Schouten-Nijenhuis  bracket.
\end{dfn}

 The space $ ( C^{\infty}(M), \lbrace .,.\rbrace_{ \Lambda,E})$  is
a local Lie algebra in the sense of Kirillov \textcolor{red}{[\ref{ref3}]} with the Jacobi bracket
\begin{equation}
\lbrace f,g\rbrace_{ \Lambda,E}= \Lambda(df,dg)+f E g-g E f\qquad \forall f,g \in C^{\infty}(M).
\end{equation}
This Lie bracket is a Poisson bracket
if and only if the vector field
$E$ identically vanishes. 

Iglesias and Marrero have proved in \textcolor{red}{[\ref{ref1}]}  that if Lie group $\mathbb{G}  $ is a connected simply connected with Lie algebra $\mathfrak{g}$ and
the pair$((\mathfrak{g}, \phi_0), (\mathfrak{g}^*, X_0))$ is a Jacobi-Lie bialgebra then $\mathbb{G}  $ 
is  Jacobi-Lie group and
has a special Jacobi
structure.  

\begin{dfn}\label{df1}
A Jacobi-Lie bialgebra is a pair $((\mathfrak{g}, \phi_0), (\mathfrak{g}^*, X_0)),$ where $\mathfrak{g}$ is a finite dimensional real Lie algebra   with Lie bracket $[ , ]^{\mathfrak{g}}$ and  $\mathfrak{g}^*$ is dual space  of $\mathfrak{g}$    with Lie bracket $[ , ]^{\mathfrak{g}^*}$; also, $\phi_0\in\mathfrak{g}^*$ and 
  $X_0 \in \mathfrak{g}$ are 1-cocycles on $\mathfrak{g}$ and $\mathfrak{g}^*$, respectively, such that
  for all $  X, Y \in \mathfrak{g}$
  satisfying the following properties:
\begin{equation}
d_{*X_0}[X,Y]^{ \mathfrak{g}}=[X, d_{*X_0}Y]_{\phi_0}^\mathfrak{g}-[Y, d_{*X_0}X]_{\phi_0}^\mathfrak{g},
\end{equation} 
\begin{equation}
\phi_0(X_0)=0,\qquad  \qquad \quad\qquad  \qquad \quad
\end{equation}
 \begin{equation}
i_{\phi_0}(d_{*}X)+[X_0, X]^\mathfrak{g}=0,\qquad  \qquad \quad
\end{equation}
   where
   \begin{equation}
i_{\phi_0}:\wedge^{\, k}\mathfrak{g}\longrightarrow \wedge^{\, k-1}\mathfrak{g},
\qquad  P\mapsto i_{\phi_0}P;
\end{equation}
 moreover,
 $d_{*}$ is the Chevalley-Eilenberg differential 
  of $\mathfrak{g}^{*}$ and $d_{*X_0}$ is the $X_0$-differential of
$\mathfrak{g}^*$
 as well as  the operation $ [ , ]_{\phi_0}^{\mathfrak{g}} $ is  the $\phi_{0}$-Schouten-Nijenhuis bracket 
\textcolor{red}{ [\ref{ref1}]}.  
\end{dfn}
\begin{thm}\label{th1}
Let the pair $((\mathfrak{g}, \phi_0), (\mathfrak{g}^*, X_0))$ be a Jacobi-Lie bialgebra and  Lie group
$\mathbb{G}  $ be a connected
simply connected  admitting Lie algebra $\mathfrak{g}$. Then, there exists a unique multiplicative
function $\sigma : \mathbb{G}   \longrightarrow \mathbb{R}$ and a unique $\sigma$-multiplicative  bivector
$ \Lambda$ on $\mathbb{G}  $ satisfying $ (d\sigma)(e) = \phi_0$
and the intrinsic derivative of  bivector at $e$ is $-d_{*X_0}$, i.e.,
$d_{e}\Lambda=-d_{*X_0}.$
 Furthermore, the following identity holds
\begin{equation}
{\Lambda}^\#(d\sigma) =\tilde{X}_0- e^{-\sigma}\bar{X }_0, 
\end{equation}
where $E = - \tilde{X }_0$
and the pair $ (\Lambda, E)$ is a Jacobi structure on $\mathbb{G}  $;
in addition,
 both
$\tilde{X}$ and $ \bar{X }$  are the  right and  left invariant vector field such that
$\tilde{X}_0(e)= \bar{X }_0(e)=X _0$
\textcolor{red}{ [\ref{ref1}]}.
\end{thm}

\begin{dfn}\label{}
A Coboundary Jacobi-Lie bialgebra is a Jacobi-Lie bialgebra 
such that $d_{*X_0}$ is a 1-coboundary that is, there exists $r \in \wedge^{2}\mathfrak{g}$  satisfying that
  \begin{equation}
  d_{*X_0}X=ad_{(\phi_0,1)}(X)(r) \quad \textit{for all}  \quad X\in \mathfrak{g}
  \end{equation}
  (for more details, see \textcolor{red}{[\ref{ref1}]}).
\end{dfn}
A Jacobi structure on $\mathbb{G}  $  was determined by Iglesias and Marrero in \textcolor{red}{[\ref{ref1}]}, as follows:
\begin{equation}\label{IG}
\Lambda=\tilde{r}-e^{-\sigma}\bar{r}, \qquad    E=-\tilde{X}_0,
\end{equation}
Where both
$\tilde{r}$ and $\bar{r}$ are
  right and  left invariant 2-vector on
the Lie group $\mathbb{G}  $
as well as $\tilde{X}_0$ is right invariant vector field on $\mathbb{G}  $.
 Furthermore, 
$ (d\sigma)(e) = \phi_0$.

 The relation (\ref{IG})
 can
be expressed in terms of coordinate on M as follows 
    \textcolor{red}{ [\ref{ref2}]}:
\begin{equation}\label{N}
\Lambda=\dfrac{1}{2}r^{ij}(X_i^{R\mu}X_j^{R\nu}-e^{-\sigma}X_i^{L\mu}X_j^{L\nu})\partial_\mu\wedge\partial_\nu, 
\end{equation}
and
\begin{equation}\label{E}
E=-\alpha^i X_i^{R\mu}\partial_\mu,\qquad\qquad\qquad\qquad\qquad\qquad
\end{equation}
Where  both $X_{i}^R$
 and $X_{i}^L$ are the $i$th  coordinates of the    right and  left invariant vector fields on
the Lie group $\mathbb{G}  $. Moreover, $r^{ij}$ is calculated from the relation $r = \frac{1}{2}r^{ij}X_i \wedge Xj$
, and  multiplicative
function $\sigma : \mathbb{G}   \longrightarrow \mathbb{R}$  is defined as
$ (d\sigma)(e) = \phi_0$
as well as
$\alpha^{i}$ is obtained   from the relation
$X_0=\alpha^{i}X_{i}$
where 
$\lbrace X_i \rbrace $ is
 the basis of the Lie algebra
 $\mathfrak{g}$.
Now, using
the results in \textcolor{red}{ [\ref{ref2}]}   and relations
 (\ref{N})  and (\ref{E}),  one can calculate  the vector field $E$ and 2-vector $\Lambda$  related to real  two- and three- dimensional   Jacobi-Lie bialgebras. The results are listed  in Tables 1-3.
 The first column gives the names of the real two- and three-dimensional  Jacobi-Lie bialgebras according to \textcolor{red}{[\ref{ref2}]}, and the second column gives the
Jacobi  structure on Lie 
group $\mathbb{G}  $ related to Lie
algebra $\mathfrak{g}$.
Note
that in \textcolor{red}{ [\ref{ref2}]} the bi-r-matrix Jacobi-Lie bialgebras as
Jacobi-Lie bialgebras being coboundary in two directions, that is,
$( (\mathfrak{g},\phi_0),(\mathfrak{g^*}, X_0)) $ and
 $((\mathfrak{g^*}, X_0), (\mathfrak{g},\phi_0) )$ having classical r-matrices r and $ \tilde{r}$   have  been classified.  Nevertheless,  they
classified coboundary Jacobi-Lie bialgebras having r or  $\tilde{r}.$
\subsection{ Jacobi-Lie  Hamiltonian systems  }
 Jacobi-Lie  Hamiltonian systems
 were introduced in \textcolor{red}{[\ref{ref9}]}.
In other words,
 Lie systems possessing a Vessiot-Guldberg Lie algebra of Hamiltonian functions with
respect to a Jacobi structure. 

One of the most useful constructions on  Jacobi manifolds is in a sense analogue of the gradient, defined as follows. 
\begin{dfn}
A  vector field X on the manifold M  with the
Jacobi structure
is said to be  a Hamiltonian 
if there exists a function f,  referred  to as Hamiltonian, such that
\begin{equation}\label{rela1}
X_f=[\Lambda, f]+f E= {\Lambda}^{\#}(d f)+f E,
\end{equation}
where
the bracket is that of Schouten-Nijenhuis  bracket\textcolor{red}{[\ref{ref9}]}.
\end{dfn}

 If $Ef=0$ (that is,
the derivative of the  function $f$ in the direction of the vector field $E$ is equal
to zero), then the  function $f$ is called a good Hamiltonian function and  $X_f$
a good Hamiltonian vector field \textcolor{red}{[\ref{ref9}]}.

 Let $f$ be a   smooth function  on the Jacobi manifold $M$.
There exists a unique vector field $X_f$ on $M$,  referred to as
Hamiltonian vector field associated with $ f,$
such that the following equality is satisfied
\begin{equation}
X_{f}g=\lbrace f,g\rbrace_{ \Lambda,E}\qquad \forall \, g\in C^{\infty}(M).
\end{equation}

However,  a vector field
$X_f$
can possess several  Hamiltonian functions.
Obviously,
$( Ham(M,\Lambda, E), [., .])$ is a Lie algebra,
where 
the bracket is that of vector fields  bracket
and
$Ham (M,\Lambda, E)$ is the space of Hamiltonian vector fields  of Jacobi manifold.

\begin{dfn}\label{df2}
A Jacobi-Lie system  is a  quadruple $(M, \Lambda, E, X)$ where  $(M, \Lambda, E)$
is a Jacobi manifold and a Lie system $X$  such that  $\mathfrak{g}^X\subset Ham(M, \Lambda, E)$ \textcolor{red}{[\ref{ref9}]}.
\end{dfn}
\begin{dfn}\label{}
A Jacobi-Lie  Hamiltonian system  is a  quadruple $(M, \Lambda, E, f)$ where  $(M, \Lambda, E)$
is a Jacobi manifold
 and $ f:  \mathbb{R}     \times M \longrightarrow M, \quad
 (t, x) \mapsto f_{t} (x)$
  is a
t-dependent function and $Lie (\lbrace  f_{t} \rbrace_{t\in  \mathbb{R} }, \lbrace ., .\rbrace_{\Lambda,E}) $ is finite-dimensional. The vector field X  on M is said to be Jacobi-Lie Hamiltonian system $(M, \Lambda, E, f)$
if $ X_{f_{t}}$  is a Hamiltonian vector field with Hamiltonian function $f_t$ (relative to
Jacobi manifold) 
$\forall t \in \mathbb{R}$ 
\textcolor{red}{[\ref{ref9}]}.
\end{dfn}
\begin{thm}\label{}
If $(M, \Lambda, E, f)$ is a Jacobi-Lie Hamiltonian system, then the system
X of the form
 $X_t = X_{f_t}, \forall  t \in \mathbb{R},$ 
 is a Jacobi-Lie system $(M, \Lambda , E, X).$ If
X is a Lie system whose $\lbrace  X_{t} \rbrace_{t\in \mathbb{R}}$ are good Hamiltonian vector fields, then X possesses
a Jacobi-Lie Hamiltonian \textcolor{red}{ [\ref{ref9}]}.
\end{thm}
\section{ Jacobi-Lie Hamiltonian systems on
 real low-dimensional 
Jacobi-Lie  groups   }
Now, we consider some Jacobi structures obtained by using the  real two and three dimensional
Jacobi-Lie groups
related to
Jacobi-Lie 
 bialgebras (see table 1, 3 below). 
  In these examples, we consider the Lie group $\mathbb{G}  $ related to the Jacobi-Lie bialgebra
 $( (\mathfrak{g},\phi_0),(\mathfrak{g^*}, X_0))$.
For
this propose we use the formalisms mentioned in the previous section for calculation of 
 Jacobi-Lie  Hamiltonian systems
on  real low dimensional Jacobi-Lie   groups. \\

{\bf Example 1.}
Real two-dimensional bi-r-matrix Jacobi-Lie bialgebra
$(( A_2,b\tilde{X}^2),(A_2.i, -bX_1)) $\textcolor{red}{[\ref{ref2}]}:\\
Consider
  the  Lie group 
$\mathbb{A}_2$  (with the coordinates $x, y$) related to Lie algebra
$A_2$.
For this example, the
Jacobi structure and Reeb vector have the following forms (see table 1):
\begin{equation}\label{A}
\Lambda_{\mathbb{A}_2}=(1-e^{-(b+1)y})\partial _x\wedge \partial _y,\qquad\qquad 
E_{\mathbb{A}_2}=b \partial _x.
\end{equation}
 Simple calculations show that they satisfying in
\begin{equation*}
[\Lambda_{\mathbb{A}_2}, \;\Lambda_{\mathbb{A}_2}]=0=2E_{\mathbb{A}_2}\wedge \Lambda_{\mathbb{A}_2},\qquad\qquad  [E_{\mathbb{A}_2},\; \Lambda_{\mathbb{A}_2}]=0,
\end{equation*}
where [.,.] is the Schouten-Nijenhuis bracket.
Thus, $( \mathbb{A}_2,  \Lambda_{\mathbb{A}_2}, E_{\mathbb{A}_2} )$ 
is a Jacobi manifold.

It is easy to check
 that 
\begin{equation*}
X_1^H
=
b \,e^{\frac{-x}{b}}\,\partial_x+
\dfrac{(-1+e^{-(1+b)y})e^{\frac{-x}{b}}}{b}\,\partial_y,
\qquad
X_2^H
=
b\,\partial_x,
\end{equation*}
span Lie algebra $A_2\, (i.e.,  [X_1^H, X_2^H]=X_1^H)$ of Hamiltonian vector fields
on $ \mathbb{A}_2.$

 Consider  now the
system on  $ \mathbb{A}_2$
 defined by 
\begin{equation}\label{000}
\frac{d\mathfrak{\alpha}_{2}}{dt}=\sum_{i=1}^2 b_{i}(t) X_i^H(\mathfrak{\alpha}_{2}),\qquad\qquad \forall \mathfrak{\alpha}_{2}
\in  \mathbb{A}_2,
\end{equation}
for arbitrary t-dependent functions $ b_{i}(t).$

 $X^{\mathbb{A}_2}$ is a Lie system
since the associated  t-dependent vector field
$X^{\mathbb{A}_2}=\sum_{i=1}^2 b_{i}(t) X_i^H$ takes values  in the Lie algebra
 $[X_1^H,X_2^H]=X_1^H$,namely,   Lie algebra $A_2$.

We now illustrate that the Lie system (\ref{000})
 is a Jacobi-Lie 
   system.

Indeed, $X_1^H$ and $ X_2^H$  are  Hamiltonian vector fields
with respect to 
 $( \mathbb{A}_2,  \Lambda_{\mathbb{A}_2}, E_{\mathbb{A}_2} )$
  with    Hamiltonian 
 functions     
 $f_1=e^{\frac{-x}{b}}$
and
$f_2=1$, respectively ( i.e.,
$
X_i^H= {\Lambda}^{\#}(d f_i)+f_i E)
$;
 consequently,
 $( \mathbb{A}_2,  \Lambda_{\mathbb{A}_2}, E_{\mathbb{A}_2},X^{\mathbb{A}_2 })$ 
 is a Jacobi-Lie 
  system.

Since
 $f=\sum_{i=1}^2 b_{i}(t) f_i=b_1(t)e^{\frac{-x}{b}} +b_2(t) $ 
is a  Hamiltonian 
 function of
$X^{\mathbb{A}_2}$ for every $t\in \mathbb{R}$ and  the
functions
$f_1$ and $f_2$ 
 satisfy the commutation
relation 
$\lbrace f_1,f_2\rbrace_{\Lambda_{\mathbb{A}_2, E_{\mathbb{A}_2}}}= f_1,$
then
 the functions $\lbrace f_t\rbrace_{t\in \mathbb{R}}$ span a finite-dimensional real  Lie algebra
 of functions
  with respect to the Lie bracket  induced by
 (\ref{A}). Consequently, $X^{\mathbb{A}_2}$
 admits a Jacobi-Lie 
 Hamiltonian  
 system
  $( \mathbb{A}_2,  \Lambda_{\mathbb{A}_2}, E_{\mathbb{A}_2},f ).$ \\

{\bf Example 2.}
Real three-dimensional bi-r-matrix Jacobi-Lie bialgebra
$(( \rom{3},-b\tilde{X}^2+b\tilde{X}^3),(\rom{3}.iv,b X_1)) $ \textcolor{red}{[\ref{ref2}]}:

Consider
  the  Lie group 
$\mathbb{\rom{3}}$  (with the coordinates $x, y, z$) related to Lie algebra
$\rom{3}$.
For this example, the
Jacobi structure and Reeb vector have the following forms (see table 2):
\begin{equation}\label{sestem1}\
\Lambda_{\mathbb{\rom{3}}}=\frac{1}{2}(1-e^{b(y-z)})\partial _x\wedge \partial _y-\frac{1}{2}(1-e^{b(y-z)})\partial _x\wedge \partial _z+(y+z)e^{b(y-z)}\partial _y\wedge \partial _z,\qquad
E_{\mathbb{\rom{3}}}=-b\partial _x.
\end{equation}
 Obviously,
 $
[\Lambda_{\mathbb{\rom{3}}},\; \Lambda_{\mathbb{\rom{3}}}]=
-2 b (y+z)e^{b(y-z)}\partial _x\wedge \partial _y\wedge\partial _z
=2E_{\mathbb{\rom{3}}}\wedge \Lambda_{\mathbb{\rom{3}}}$
\quad and 
\quad $ [E_{\mathbb{\rom{3}}}  ,\; \Lambda_{\mathbb{\rom{3}}}]=0.$

As a result,
$( \mathbb{\rom{3}},  \Lambda_{\mathbb{\rom{3}}}, E_{\mathbb{\rom{3}}} )$ 
is a Jacobi manifold.
It is straightforward to verify
that
 the   Lie algebra, $\rom 2,$ of Hamiltonian vector fields
on $ \mathbb{\rom 3}$ is
spanned by
\begin{equation*}
X_1^H=-b\,\partial_x,\quad 
X_2^H
=
\Big(-\frac{1}{2}+\frac{1}{2}e^{b(y-z)}-by\Big)\partial_x
+
(y+z)e^{b(y-z)}\,\partial_z,
\end{equation*}
$$
X_3^H=
Ei(1,-b(y+z))e^{-b(y+z)} b \,\partial_x
-\partial_y
+\left( 2\,b \left( y+z \right) {\it Ei} \left( 1,-b \left( y+z
 \right)  \right) {e}^{-b
 \left( y+z \right)} +1\right)  
\partial_z
$$
 where
 ${\it Ei} \left( 1,-b \left( y+z \right)  \right)=\int _{-b \left( y+z \right) }^{\infty }\!{\frac {{{\rm e}^{-x}}}{x}}{
dx}$

The system on $ \mathbb{\rom 3}$ can be defined as
\begin{equation}
\frac{d\mathbb{\gamma}}{dt}=\sum_{i=1}^3 b_{i}(t) X_i^H(\mathbb{\gamma}),\qquad\qquad \forall \mathbb{\gamma}\in  \mathbb{\rom 3},
\end{equation}
for arbitrary t-dependent functions $ b_{i}(t).$

So that the associated  t-dependent vector field
$X^{\mathbb{\rom3}}=\sum_{i=1}^3 b_{i}(t) X_i^H$
 takes values  in the Lie algebra $\rom{2}$ 
,that  is,\;$[X_2^H, X_3^H]=X_1^H$
 ; so  $X^{\mathbb{\rom3}}$ is a Lie system.
 
 We now manifest that
 $( \mathbb{\rom3},  \Lambda_{\mathbb{\rom3}}, E_{\mathbb{\rom3}}, X^{\mathbb{\rom3}})$ 
 is a Jacobi-Lie  system. Infact, $X_1^H$ and $ X_2^H$,and $X_3^H$   are  Hamiltonian vector fields
relative to 
$( \mathbb{\rom3},  \;\Lambda_{\mathbb{\rom3}}, \;E_{\mathbb{\rom3}} )$
  with \textcolor{red}{good}  Hamiltonian 
 functions     
 $f_1=1$
and
$f_2=y$, and $f_3=-e^{-2by}Ei(1,\;-b(y+z)$ respectively  ( i.e.,
$
X_i^H= {\Lambda}^{\#}(d f_i)+f_i E )
$;
 subsequently,
  $( \mathbb{\rom3},\;  \Lambda_{\mathbb{\rom3}},\; E_{\mathbb{\rom3}},X^{\mathbb{\rom3} })$ 
 is a Jacobi-Lie  system.
Because
 $f=\sum_{i=1}^2 b_{i}(t) f_i=b_1(t) +b_2(t) y-b_3(t)e^{-2by}Ei(1,-b(y+z)$ 
is a  Hamiltonian 
 function of
$X^{\mathbb{\rom3}}$ for every $t\in \mathbb{R} $, and the
functions
$f_1,f_2$ 
and
$f_3$ satisfy the commutation
relations 
$\lbrace f_2,f_3\rbrace_{\Lambda_{\mathbb{\rom3}, \, E_{\mathbb{\rom3}}} }= f_1$,
then
 the functions $\lbrace f_t\rbrace_{t\in \mathbb{R}}$ span a finite-dimensional real  Lie algebra
 of functions
  with respect to the Lie bracket  induced by
  (\ref{sestem1}). Consequently, $X^{\mathbb{\rom3}}$
 admits a Jacobi-Lie  Hamiltonian  system
  $( \mathbb{\rom3},  \Lambda_{\mathbb{\rom3}}, E_{\mathbb{\rom3}},f).$ \\
  
{\bf Example 3.}
Real three-dimensional bi-r-matrix Jacobi-Lie bialgebra
$(( \rom{3},0),(\rom{3}.v,\frac{1}{2}X_2-\frac{1}{2}X_3)) $  \textcolor{red}{[\ref{ref2}]}:

Consider
  the  Lie group 
$\mathbb{\rom{3}}$  (with the coordinates $x, y, z$) related to Lie algebra
$\rom{3}$.
For this example, the
Jacobi structure and Reeb vector have the following forms (see table 2):

\begin{equation}\label{33333}
\Lambda_ \mathbb{\rom{3}}=\dfrac{e^{2x}-1}{4}\partial _x\wedge \partial _y+\dfrac{e^{2x}-1}{4}\partial _x\wedge \partial _z-\frac{1}{2}(y+z+1-e^{2x})\partial _y\wedge \partial _z,\qquad
E_ \mathbb{\rom{3}}=-\frac{1}{2}\partial _y+\frac{1}{2}\partial _z,
\end{equation}
one can show that
 $$
[\Lambda_{\mathbb{\rom{3}}},\; \Lambda_{\mathbb{\rom{3}}}]=\dfrac{e^{2x}-1}{2}\partial _x\wedge \partial _y \wedge \partial _z =2E_{\mathbb{\rom{3}}}\wedge \Lambda_{\mathbb{\rom{3}}},\quad  \quad  [E_{\mathbb{\rom{3}}} ,\;\Lambda_{\mathbb{\rom{3}}}]=0.$$
So
$( \mathbb{\rom{3}},  \Lambda_{\mathbb{\rom{3}}}, E_{\mathbb{\rom{3}}} )$
is a Jacobi manifold.
A  simple calculation shows that 

$
X_1^H
=
-\frac{1}{2}\, \left( 2\,x-3\,y-z-2 \right)  \left( -1+{e}^{2\,x} \right) 
\partial_x
+\Big(
\frac{3}{2}\,{e}^{2\,x}y+\frac{1}{2}\,{e}^{2\,x}z+{e}^{2\,x}+\frac{y}{2}-\frac{z}{2}-1-{e}^{2\,x}x-x
\Big)\partial_y+
\Big(
-\frac{3}{2}\,{e}^{2\,x}y-\frac{1}{2}\,{e}^{2\,x}z-{e}^{2\,x}+\frac{3}{2}\,y+\frac{5}{2}\,z+1+{e}^{2\,
x}x+{y}^{2}+2\,yz+{z}^{2}+x
\Big)\partial_z,
\,
X_2^H
=
\Big(
\frac{1}{2}-\frac{1}{2}e^{2x}\Big)\partial_x
-\frac{1}{2}e^{2x}\,\partial_y
+
\frac{1}{2}\,e^{2x}\,\partial_z
$
and
$
X_3^H
=
-\frac{1}{2}\partial_y
+\frac{1}{2}\partial_z
$
span the   Lie algebra $\rom 3$ of Hamiltonian vector fields
on $ \mathbb{\rom 3}.$

The system on $ \mathbb{\rom 3}$ can be written as
\begin{equation}
\frac{d\mathbb{\gamma}}{dt}=\sum_{i=1}^3 b_{i}(t) X_i^H(\mathbb{\gamma}),\qquad\qquad \forall \mathbb{\gamma}\in  \mathbb{\rom3},
\end{equation}
for arbitrary t-dependent functions $ b_{i}(t).$

Since the associated  t-dependent vector field
$X^{\mathbb{\rom3}}=\sum_{i=1}^3 b_{i}(t) X_i^H$ takes values  in the Lie algebra $\rom{3}$ 
,that  is,\;$[X_1^H, X_2^H]=-(X_2^H+X_3^H),  [X_1^H, X_3^H]=-(X_2^H+X_3^H)$
 , then $X^{\mathbb{\rom3}}$ is a Lie system.
 
 We now prove that
 $( \mathbb{\rom3},  \Lambda_{\mathbb{\rom3}}, E_{\mathbb{\rom3}}, X^{\mathbb{\rom3}})$ 
 is a Jacobi-Lie  system. As a matter of fact, $X_1^H$ and $ X_2^H$,and $X_3^H$   are Hamiltonian vector fields
relative to 
 $( \mathbb{\rom3},  \Lambda_{\mathbb{\rom3}}, E_{\mathbb{\rom3}} )$
  with   Hamiltonian 
 functions     
 $f_1=2\, \left( y+z+2 \right)  \left( -y+x \right)$
and
$f_2=1+y+z$, and $f_3=1$ respectively ( i.e.,
$
X_i^H= {\Lambda}^{\#}(d f_i)+f_i E );
$ subsequently,
  $( \mathbb{\rom3},  \Lambda_{\mathbb{\rom3}}, E_{\mathbb{\rom3}},,X^{\mathbb{\rom3} })$ 
 is a Jacobi-Lie  system.
 
Using the Lie bracket  induced by
 $ \Lambda_{ \mathbb{\rom3}}$  and $E_{ \mathbb{\rom3}}$ of Lie group $ \mathbb{\rom3}$, 
 we can write
$\lbrace f_1,f_2\rbrace_{\Lambda_{ \mathbb{\rom3}, \, E_{ \mathbb{\rom3}} }}=- f_2-f_3,\quad
\lbrace f_1,f_3\rbrace_{\Lambda_{ \mathbb{\rom3}, \, E_{ \mathbb{\rom3}} }}=-f_2- f_3.$
Therefore,
$( \mathbb{\rom3},  \Lambda_{\mathbb{\rom3}}, E_{\mathbb{\rom3}},f=\sum_{i=1}^3 b_{i}(t) f_i )$
for   $X^{\mathbb{\rom3}}$ 
 is a Jacobi-Lie   Hamiltonian system,
 where  $ \Lambda_{ \mathbb{\rom3}}$  and $E_{ \mathbb{\rom3}}$
 are those appearing in (\ref{33333}).\\

{\bf Example 4.}
Real three-dimensional bi-r-matrix Jacobi-Lie bialgebra
$(( \rom{4},-\tilde{X}^1),(\rom{3}.vi,-X_2-X_3)) $  \textcolor{red}{ [\ref{ref2}]}:

Consider
  the  Lie group 
$\mathbb{\rom{4}}$  (with the coordinates $x, y, z$) related to Lie algebra
$\rom{4}$.
For this example, the
Jacobi structure and Reeb vector have the following forms (see table 2):

\begin{equation}\label{88888}
\Lambda_{\mathbb{\rom{4}}}=-xe^{x}\partial _x\wedge \partial _z+e^{x}(z-y-1+e^{x})\partial _y\wedge \partial _z,\quad
E_{\mathbb{\rom{4}}}=e^{x}\partial _y+e^{x}(1-x)\partial _z.
\end{equation}

Then one can show  that 
 $$
[\Lambda_{\mathbb{\rom{4}}}, \;\Lambda_{\mathbb{\rom{4}}}]=2\,x{e}^{2\,x}\partial _x\wedge \partial_y\wedge\partial _z=2E_{\mathbb{\rom{4}}}\wedge \Lambda_{\mathbb{\rom{4}}},\quad  \quad  [E_{\mathbb{\rom{4}}},\; \Lambda_{\mathbb{\rom{4}}}]=0.$$
so,
$( \mathbb{\rom{4}},  \Lambda_{\mathbb{\rom{4}}}, E_{\mathbb{\rom{4}}} )$
is a Jacobi manifold.
A short calculation shows that
$$
X_1^H
=
-\,\partial _x+
\frac {y-1+{e}^{x}}{x}\,\partial _y
+{\frac {2\,{e}^{x}+2\,y-2}{x}}\,
\partial _z,\quad
X_2^H
=
e^{x}\,\partial _y
+e^{x}(1-x)\,\partial _z $$
and
$$
X_3^H
=
-\frac{e^{x}}{x}\,\partial _y
+\frac {{e}^{x} \left( x-2 \right) }{x}\,
\partial _z$$

span the   Lie algebra $,\rom 4,$ of Hamiltonian vector fields
on $ \mathbb{\rom 4}.$
The system on $ \mathbb{\rom 4}$ can be considered as

\begin{equation}
\frac{d\mathbb{\delta}}{dt}=\sum_{i=1}^3 b_{i}(t) X_i^H(\mathbb{\delta}),\qquad\qquad \forall \mathbb{\delta}\in  \mathbb{\rom 4},
\end{equation}
for arbitrary t-dependent functions $ b_{i}(t).$
Since the associated  t-dependent vector field
$X^{\mathbb{\rom4}}=\sum_{i=1}^3 b_{i}(t) X_i^H$ takes values  in the Lie algebra $\rom{4}$ 
,that  is, $[X_1^H, X_2^H]=-(X_2^H-X_3^H),  [X_1^H, X_3^H]=-X_3^H$
 ; then, $X^{\mathbb{\rom4}}$ is a Lie system.
 We now exhibits that
 $( \mathbb{\rom4},  \Lambda_{\mathbb{\rom4}}, E_{\mathbb{\rom4}},  X^{\mathbb{\rom4} })$ 
 is a Jacobi-Lie   system. Actually, $X_1^H$ and $ X_2^H$,and $X_3^H$   are Hamiltonian vector fields
relative to 
 $( \mathbb{\rom4},  \Lambda_{\mathbb{\rom4}}, E_{\mathbb{\rom4}} )$
  with   Hamiltonian 
 functions     
 $f_1={\frac { \left( 2\,y-z \right) {e}^{-x}}{x}}$
and
$f_2=1$, and $f_3=-\frac{1}{x}$ respectively ( i.e.,
$
X_i^H= {\Lambda}^{\#}(d f_i)+f_i E)
$; subsequently
  $( \mathbb{\rom4},  \Lambda_{\mathbb{\rom4}}, E_{\mathbb{\rom4}}, X^{\mathbb{\rom4} })$ 
 is a Jacobi-Lie system.
 
 Using the Lie bracket  induced by
 $ \Lambda_{ \mathbb{\rom4}}$  and $E_{ \mathbb{\rom4}}$ of Lie group $ \mathbb{\rom4}$, the
functions
$f_1,f_2$ 
and
$f_3$ satisfy the commutation
relations 
$\lbrace f_1,f_2\rbrace_{\Lambda_{ \mathbb{\rom4}, \, E_{ \mathbb{\rom4}} }}=- f_2+f_3,\quad
\lbrace f_1,f_3\rbrace_{\Lambda_{ \mathbb{\rom4}, \, E_{ \mathbb{\rom4}} }}=- f_3.$
Therefore,
$( \mathbb{\rom4},  \Lambda_{\mathbb{\rom4}}, E_{\mathbb{\rom3}},f=\sum_{i=1}^3 b_{i}(t)  )$
for   $X^{ \mathbb{\rom4}}$ 
 is a Jacobi-Lie   Hamiltonian system,
  where  $ \Lambda_{ \mathbb{\rom4}}$  and $E_{ \mathbb{\rom4}}$
 are those appearing in (\ref{88888}).
 
 \section{Lie symmetry for  Jacobi-Lie Hamiltonian systems}
 
 We now present some examples of Jacobi-Lie Hamiltonian systems on three-dimensional real Jacobi-Lie groups
whose distribution associated with their systems is of dimension two. Then, there exists a constant of motion
for their systems. We obtain a  time-independent Lie symmetry \textcolor{red}{[\ref{ref7}]} for their systems to illustrate our procedure.


 Let $X$ be a 
  t-dependent vector field  on M, the associated distribution of  $X$  is the generalised distribution $\Delta ^X$ on M spanned by the vector fields of $ \mathfrak{g}^ X.$ In other words,
\[\Delta _x^X = \lbrace{Z_x} \;\vert \;Z \in {\mathfrak{g}^X}\rbrace \subset {T_x}M\]
and the associated co-distribution  of  $X$ is the generalised co-distribution 
${(\Delta ^X})^{\perp}$
on M of the form
\[({\Delta _x^X})^{\perp}= \lbrace\nu  \in T_x^*M\;|\;\nu ({Y_x}) = 0,\forall {Y_x} \in \Delta_x^X\rbrace \subset T_x^*M.\]
where ${(\Delta_x ^X})^{\perp}$ is the annihilator of $\Delta _x^X$.

The function
$
 \rho ^X:  M \longrightarrow \mathbb{N}\cup \lbrace 0 \rbrace, \quad x\mapsto dim \,\Delta _x^X
$
 is a lower semicontinuous at $x$ since it cannot decrease  in a neighbourhood of $x$.
 In addition, $ \rho ^X(x)$
 is constant on the
connected components of dense and  an open subset $U^ X$ of $M$  ( cf. \textcolor{red}{[\ref{11}}, p. 19] ),
where $\Delta ^X$ becomes a regular involutive distribution. Also,
$(\Delta _x^X)^\perp$ becomes a
regular co-distribution on each connected component 
since 
$dim   (\Delta _x^X)^\perp= dim M - \rho^X (x)$.
\begin{pro}\label{}
A function $h : U^X \longrightarrow  \mathbb{R}$ is a local t-independent constant of motion for  a t-dependent vector field X if
and only if $dh \in (\Delta _x^X)^\perp \vert_{U^X}$\,\textcolor{red}{ [\ref{ref7}]}. 
\end{pro}
\begin{dfn}\label{}
 Let X be a Jacobi-Lie system  with  a Jacobi-Lie Hamiltonian structure $(M, \Lambda, E, f)$.
 Then,
one can define  its symmetry distribution 
as follows:
$(S^X_{ \Lambda,E})_x={\Lambda}^{\#}(d h_i)+h_i E \in T_xM, $
where  $dh_i \in (\Delta _x^X)^\perp \vert_U.$ 
\end{dfn}
Now using the  symmetry distribution, we study the t-independent
Lie symmetries of Jacobi-Lie Hamiltonian systems on real
low dimensional Jacobi-Lie groups.

The following proposition is the same as the proposition 1 in \textcolor{red}{[\ref{ref9}]},  except that Hamiltonian functions are not
necessarily a good.

\begin{pro}\label{}
 Let X be a Jacobi-Lie system   possessing   a Jacobi-Lie Hamiltonian structure $(M, \Lambda, E, f)$.
The 
  smooth function h on the Jacobi manifold M
is a constant of motion for X if and only if it  commutes with all
elements of 
$Lie (\lbrace  f_{t} \rbrace_{t\in  \mathbb{R} }, \lbrace ., .\rbrace_{\Lambda,E}) $ relative to $\lbrace ., .\rbrace_{\Lambda,E}$.
\end{pro}
\begin{proof}
Let  $h  $   be a t-independent constant of motion for X,  i.e. 
\begin{equation}\label{jaco}
\lbrace  f_t, h \rbrace_{ \Lambda,E}=X_{ f_{t}} h=0, \quad \forall t\in
\mathbb{R}.
\end{equation}
Using (\ref{jaco}) and the Jacobi identity, we get

$\lbrace h, \lbrace  f_t , f_{t^\prime} \rbrace_{ \Lambda,E}\rbrace_{ \Lambda,E}=
\lbrace f_t, \lbrace    h, f_{t^\prime} \rbrace_{ \Lambda,E}\rbrace_{ \Lambda,E}
+\lbrace\lbrace  h , f_t  \rbrace_{ \Lambda,E}, f_{t^\prime}\rbrace_{ \Lambda,E} =0\quad \forall t, t^\prime \in
\mathbb{R}.$

By the inductive procedure, 
 $h$  commutes with the whole  successive Jacobi brackets of  elements of
$\lbrace  f_{t} \rbrace_{t\in  \mathbb{R} }$ relative to $\lbrace ., .\rbrace_{\Lambda,E}$ and their linear combinations.
Since $\lbrace  f_{t} \rbrace_{t\in  \mathbb{R} }$ span $Lie (\lbrace  f_{t} \rbrace_{t\in  \mathbb{R} }, \lbrace ., .\rbrace_{\Lambda,E}) $, we obtain that h 
commutes with $Lie (\lbrace  f_{t} \rbrace_{t\in  \mathbb{R} }, \lbrace ., .\rbrace_{\Lambda,E}) $ relative to $\lbrace ., .\rbrace_{\Lambda,E}$.\\
Let us prove the converse. If $h$  commutes with    $Lie (\lbrace  f_{t} \rbrace_{t\in  \mathbb{R} }, \lbrace ., .\rbrace_{\Lambda,E}) $ relative to $\lbrace ., .\rbrace_{\Lambda,E}$, it commutes with the elements $\lbrace  f_{t} \rbrace_{t\in  \mathbb{R} }$
 relative to $\lbrace ., .\rbrace_{\Lambda,E}$.
 By applying (\ref{jaco}), the function h  is  a constant of motion of X.
\end{proof}\\

\begin{lem}\label{lm1}
The mappinq
$$\varphi:( C^{\infty}(M), \lbrace .,.\rbrace_{ \Lambda,E})\longrightarrow (Ham (M,\Lambda, E), [., .])$$
is a homomorphism Lie algebras, i.e.
$\varphi{\lbrace f, g\rbrace_{ \Lambda,E}}=[X_f, X_g]$.
\end{lem}

\begin{pro}\label{pro}
 Let X be a Jacobi-Lie system   admitting a Jacobi-Lie Hamiltonian structure $(M, \Lambda, E, f)$.
 If h is a t-independent constant of motion for X, then $X_h={\Lambda}^{\#}(d h)+h E$  is a t-independent Lie symmetry of
X.
\end{pro}
\begin{proof}
 In view of Lemma
   \ref{lm1}, we have\\
$
[ X_h, X_{f_{t}}]=[{\Lambda}^{\#}(d h)+h E,{\Lambda}^{\#}(d f_t)+f_t E]=\varphi\lbrace  h, f_t \rbrace_{ \Lambda,E}=-\varphi\lbrace  f_t, h \rbrace_{ \Lambda,E}=-\varphi(X_t h)=0,  \forall t\in
\mathbb{R}.
$
\end{proof}\\


{\bf Example 1.}
Real three-dimensional bi-r-matrix Jacobi-Lie bialgebra
$(( \rom{3},-b\tilde{X}^2+b\tilde{X}^3),(\rom{3}.iv,b X_1)) $ \textcolor{red}{[\ref{ref2}]}:

Let us Consider
  the  Lie group 
$\mathbb{\rom{3}}$  related to Lie algebra
$\rom{3}$ and  the Jacobi structure given by (\ref{sestem1}).
A  simple calculation shows that

$$
X_1^H
=
-b x \partial _x+\frac{1}{2}(1-{e}^{b \left( y-z \right) }
)\partial _y-\frac{1}{2}(1-{e}^{b \left( y-z \right) }
)\partial _z
$$
and 

\qquad\qquad
$$X_2^H
= x ln(x) \partial _x-\frac{1}{2b}(1-{e}^{b \left( y-z \right) }
)(1+ln(x))\partial _y+\frac{1}{2b}(1-{e}^{b \left( y-z \right) }
)(1+ln(x))\partial _z
$$

span the   Lie algebra $,A_2,$ of Hamiltonian vector fields
on $ \mathbb{\rom 3}.$

The system on $ \mathbb{\rom 3}$ can be written as
\begin{equation}
\frac{d\mathbb{\gamma}}{dt}=\sum_{i=1}^2 b_{i}(t) X_i^H(\mathbb{\gamma}),\qquad\qquad \forall\mathbb{\gamma}\in  \mathbb{\rom3},
\end{equation}
for arbitrary t-dependent functions $ b_{i}(t).$

Since the associated  t-dependent vector field
$X^{\mathbb{\rom3}}=\sum_{i=1}^2 b_{i}(t) X_i^H$ takes values  in the Lie algebra $A_2$ 
,that  is,\;$[X_1^H, X_2^H]=X_1^H$
 , then $X^{\mathbb{\rom3}}$ is a Lie system.
 
 We now prove that
 $( \mathbb{\rom3},  \Lambda_{\mathbb{\rom3}}, E_{\mathbb{\rom3}}, X^{\mathbb{\rom3}})$ 
 is a Jacobi-Lie  system. As a matter of fact, $X_1^H$ and $ X_2^H$   are Hamiltonian vector fields
relative to 
 $( \mathbb{\rom3},  \Lambda_{\mathbb{\rom3}}, E_{\mathbb{\rom3}} )$
  with  Hamiltonian 
 functions     
 $f_1=x$
and
$f_2=-{\frac {x\ln  \left( x \right) }{b}}$ respectively. Subsequently,
  $( \mathbb{\rom3},  \Lambda_{\mathbb{\rom3}}, E_{\mathbb{\rom3}},,X^{\mathbb{\rom3} })$ 
 is a Jacobi-Lie  system.
 
Using the Lie bracket  induced by
 $ \Lambda_{ \mathbb{\rom3}}$  and $E_{ \mathbb{\rom3}}$ of Lie group $ \mathbb{\rom3}$, 
 we can write
$\lbrace f_1,f_2\rbrace_{\Lambda_{ \mathbb{\rom3},  E_{ \mathbb{\rom3}} }}= f_1;$
therefore,
$( \mathbb{\rom3},  \Lambda_{\mathbb{\rom3}}, E_{\mathbb{\rom3}},f=\sum_{i=1}^2 b_{i}(t) f_i )$
for   $X^{\mathbb{\rom3}}$ 
 is a Jacobi-Lie   Hamiltonian system.
\\ 
 It is easy  check that
$h=1-{e}^{-b \left( y-z \right) }
$
is a t-independent constant of motion.
One can  show that
$$\lbrace   h\, ,f_\alpha\rbrace_{\Lambda_{\mathbb{\rom3}}, E_{\mathbb{\rom3}}}=0,\qquad \alpha=1,2.$$
Then, the function $h$ always Jacobi commutes with the whole Lie algebra $Lie(\lbrace f_t \rbrace _{ t\in  \mathbb{R}} ,\lbrace.,.\rbrace_{\Lambda_{\mathbb{\rom3}}, E_{\mathbb{\rom3}}}),$ as expected.
 By applying
proposition \ref{pro}, $X_h={\Lambda}^{\#}(dh)+h E$   must be a Lie symmetry for this system. A short calculation shows
that
$X_h=\left( y+z \right) b\partial_y+ \left( y+z \right) b\partial_z.$
It is easy to check that  $X_h$ commutes with $X_1^H, X_2^H,$ and thus  commutes   with
every $X_{f_{t}}, $ with $t \in \mathbb{R},$ i.e. $X_h$ is a Lie symmetry for $X^{\mathbb{\rom3}}$.\\

{\bf Example 2.}
Real three-dimensional bi-r-matrix Jacobi-Lie bialgebra
$(( \rom{3},0),(\rom{3}.v,\frac{1}{2}X_2-\frac{1}{2}X_3)) $  \textcolor{red}{[\ref{ref2}]}:

Let us Consider
  the  Lie group 
$\mathbb{\rom{3}}$  related to Lie algebra
$\rom{3}$ and  the Jacobi structure given by (\ref{33333}).
A  simple calculation shows that

$$
X_1^H
=
\dfrac{1-e^{2x}}{4}\partial _x-\frac{y}{2}\partial _y-\frac{1}{2}(1+z-{e}^{2x})\partial _z
$$
and 
\qquad\qquad
$$X_2^H
=\frac{1}{2}\, \left( -1+{e}^{2\,x} \right)  \left( 1+\ln  \left( y \right) 
 \right) \partial _x
+y\ln  \left( y \right)\partial _y
+ (\left( -\ln  \left( y \right) -1 \right) {e}^{2\,x}+ \left( z+1
 \right) \ln  \left( y \right) +y+z+1) \partial _z
$$

span the   Lie algebra $,A_2,$ of Hamiltonian vector fields
on $ \mathbb{\rom 3}.$

The system on $ \mathbb{\rom 3}$ can be written as
\begin{equation}
\frac{d\mathbb{\gamma}}{dt}=\sum_{i=1}^2 b_{i}(t) X_i^H(\mathbb{\gamma}),\qquad\qquad \forall\mathbb{\gamma}\in  \mathbb{\rom3},
\end{equation}
for arbitrary t-dependent functions $ b_{i}(t).$

Since the associated  t-dependent vector field
$X^{\mathbb{\rom3}}=\sum_{i=1}^2 b_{i}(t) X_i^H$ takes values  in the Lie algebra $A_2,$ that  is,\;$[X_1^H, X_2^H]=X_1^H$
 , then $X^{\mathbb{\rom3}}$ is a Lie system.
 
 We now prove that
 $( \mathbb{\rom3},  \Lambda_{\mathbb{\rom3}}, E_{\mathbb{\rom3}}, X^{\mathbb{\rom3}})$ 
 is a Jacobi-Lie  system. As a matter of fact, $X_1^H$ and $ X_2^H$   are Hamiltonian vector fields
relative to 
 $( \mathbb{\rom3},  \Lambda_{\mathbb{\rom3}}, E_{\mathbb{\rom3}} )$
  with  Hamiltonian 
 functions     
 $f_1=y$
and
$f_2=-2 \,y \,ln(y)$ respectively. Subsequently,
  $( \mathbb{\rom3},  \Lambda_{\mathbb{\rom3}}, E_{\mathbb{\rom3}},,X^{\mathbb{\rom3} })$ 
 is a Jacobi-Lie  system.
 
Using the Lie bracket  induced by
 $ \Lambda_{ \mathbb{\rom3}}$  and $E_{ \mathbb{\rom3}}$ of Lie group $ \mathbb{\rom3}$, 
 we can write
$\lbrace f_1,f_2\rbrace_{\Lambda_{ \mathbb{\rom3},  E_{ \mathbb{\rom3}} }}= f_1;$
therefore,
$( \mathbb{\rom3},  \Lambda_{\mathbb{\rom3}}, E_{\mathbb{\rom3}},f=\sum_{i=1}^2 b_{i}(t) f_i )$
for   $X^{\mathbb{\rom3}}$ 
 is a Jacobi-Lie   Hamiltonian system.
 
 It is easy to show that
$h=\ln  \left( {e}^{2\,x}-1 \right) -{e}^{-2\,x}\ln  \left( {e}^{2\,x}-1
 \right) +y+z
$
is a t-independent constant of motion.
One can  check that
$$\lbrace   h\, ,f_\alpha\rbrace_{\Lambda_{\mathbb{\rom3}}, E_{\mathbb{\rom3}}}=0,\qquad \alpha=1,2.$$
Then, the function $h$ always Jacobi commutes with the whole Lie algebra $Lie(\lbrace f_t \rbrace _{ t\in  \mathbb{R}} ,\lbrace.,.\rbrace_{\Lambda_{\mathbb{\rom3}}, E_{\mathbb{\rom3}}}),$ as expected.

 By applying
proposition \ref{pro},
  $X_h={\Lambda}^{\#}(dh)+h E$   must be a Lie symmetry for this system. A short calculation shows
that
$X_h=\frac{1}{2}(1-e^{2 x})\partial_x+ {\frac { \left( {e}^{-2\,x}+{e}^{2\,x}-2 \right) \ln  \left( {e}^{2\,x
}-1 \right) -2\,{e}^{2\,x}+{e}^{4\,x}+1}{{e}^{2\,x}-1}}
\partial_z.$
It is easy to check that  $X_h$ commutes with $X_1^H, X_2^H,$ and hence with
every $X_{f_{t}}, $ with $t \in \mathbb{R},$ i.e. $X_h$ is a Lie symmetry for $X^{\mathbb{\rom3}}$.\\

{\bf Example 3.}
Real three-dimensional bi-r-matrix Jacobi-Lie bialgebra
$(( \rom{4},-\tilde{X}^1),(\rom{3}.vi,-X_2-X_3))${[\ref{ref2}]}:

Let us Consider
  the  Lie group 
$\mathbb{\rom{4}}$  related to Lie algebra
$\rom{4}$ and  the Jacobi structure given by (\ref{88888}).
A  simple calculation shows that

$$
X_1^H
 =({{\rm e}^{x-{e}^{-x}y}})\partial _y-{{\rm e}^{-{e}^{-x}y}}(xe^x+xy-y+z-1)\partial _z,
$$
and 
$$X_2^H
= e^x\partial _y+ e^x (x-1)\partial _z
$$

span the   Lie algebra $,A_2,$ of Hamiltonian vector fields
on $ \mathbb{\rom 4}.$

The system on $ \mathbb{\rom 4}$ can be written as
\begin{equation}
\frac{d\mathbb{\delta}}{dt}=\sum_{i=1}^2 b_{i}(t) X_i^H(\mathbb{\delta}),\qquad\qquad \forall\mathbb{\delta}\in  \mathbb{\rom4},
\end{equation}
for arbitrary t-dependent functions $ b_{i}(t).$

Since the associated  t-dependent vector field
$X^{\mathbb{\rom4}}=\sum_{i=1}^2 b_{i}(t) X_i^H$ takes values  in the Lie algebra $A_2$ 
,that  is,\;$[X_1^H, X_2^H]=X_1^H$
 , then $X^{\mathbb{\rom4}}$ is a Lie system.
 
 We now prove that
 $( \mathbb{\rom4},  \Lambda_{\mathbb{\rom4}}, E_{\mathbb{\rom4}}, X^{\mathbb{\rom4}})$ 
 is a Jacobi-Lie  system. As a matter of fact, $X_1^H$ and $ X_2^H$   are Hamiltonian vector fields
relative to 
 $( \mathbb{\rom4},  \Lambda_{\mathbb{\rom4}}, E_{\mathbb{\rom4}} )$
  with  Hamiltonian 
 functions     
 $f_1={{\rm e}^{-{e}^{-x}y}}$
and
$f_2=1 $ respectively. Subsequently,
  $( \mathbb{\rom4},  \Lambda_{\mathbb{\rom4}}, E_{\mathbb{\rom4}},,X^{\mathbb{\rom4} })$ 
 is a Jacobi-Lie  system.
 
Using the Lie bracket  induced by
 $ \Lambda_{ \mathbb{\rom4}}$  and $E_{ \mathbb{\rom4}}$ of Lie group $ \mathbb{\rom4}$, 
 we can write
$\lbrace f_1,f_2\rbrace_{\Lambda_{ \mathbb{\rom4},  E_{ \mathbb{\rom4}} }}= f_1;$
therefore,
$( \mathbb{\rom4},  \Lambda_{\mathbb{\rom4}}, E_{\mathbb{\rom4}},f=\sum_{i=1}^2 b_{i}(t) f_i )$
for   $X^{\mathbb{\rom4}}$ 
 is a Jacobi-Lie   Hamiltonian system.
\\ 
 It is easy to show that
$h={e}^{x}+ \left( x-1 \right) y+z-1$
is a t-independent constant of motion.
One can  check that
$$\lbrace   h\, ,f_\alpha\rbrace_{\Lambda_{\mathbb{\rom4}}, E_{\mathbb{\rom4}}}=0,\qquad \alpha=1,2.$$
Then, the function $h$ always Jacobi commutes with the whole Lie algebra $Lie(\lbrace f_t \rbrace _{ t\in  \mathbb{R}} ,\lbrace.,.\rbrace_{\Lambda_{\mathbb{\rom4}}, E_{\mathbb{\rom4}}}),$ as expected.

 By applying
proposition \ref{pro},
 $X_h={\Lambda}^{\#}(dh)+h E$   must be a Lie symmetry for this system. A short calculation  shows that
$X_h=xe^x\partial_x+  xye^x\partial_y -x(xye^x+e^{2x})\partial_z.$
It is easy to check that $X_h$ commutes with $X_1^H, X_2^H,$ and subsequently with
every $X_{f_{t}}, $ with $t \in \mathbb{R},$ i.e. $X_h$ is a Lie symmetry for $X^{\mathbb{\rom4}}$. \\

{\bf Example 4.}
Real three-dimensional bi-r-matrix Jacobi-Lie bialgebra
$(( \rom{6}_0,\tilde{X}^3),(\rom{3}.ix,-X_1))${[\ref{ref2}]}:

Let us Consider
  the  Lie group 
$\mathbb{\rom{6}}_0$  related to Lie algebra
$\rom{6}_0$ and  the Jacobi structure given by 
\begin{equation*}\label{}
\Lambda_{\mathbb{\rom{6}}_0}=(1+y-e^{-z})\partial _x\wedge \partial _y+e^{-z}\sinh(z)\partial _x\wedge \partial _z+(1-e^{-z}\cosh(z))\partial _y\wedge \partial _z,\quad
E_{\mathbb{\rom{6}}_0}=\partial _x.
\end{equation*}
Then one can show  that 
 $$
[\Lambda_{\mathbb{\rom{6}}_0}, \;\Lambda_{\mathbb{\rom{6}}_0}]=2(1-e^{-z}cosh(z))\partial _x\wedge \partial_y\wedge\partial _z=2E_{\mathbb{\rom{4}}}\wedge \Lambda_{\mathbb{\rom{4}}},\quad  \quad  [E_{\mathbb{\rom{6}}_0},\; \Lambda_{\mathbb{\rom{6}}_0}]=0.$$
so,
$( \mathbb{\rom{6}}_0,  \Lambda_{\mathbb{\rom{6}}_0}, E_{\mathbb{\rom{6}}_0} )$
is a Jacobi manifold.

A  simple calculation shows that
$
X_1^H
 =\partial _x
$
and 
$X_2^H
= x\partial _x+ (1+y-{e}^{-z})\partial _y+{e}^{-z}\sinh \left( z \right) \partial _z
$
span the   Lie algebra $,A_2,$ of Hamiltonian vector fields
on $ \mathbb{\rom 6}_0.$

The system on $ \mathbb{\rom 6}_0$ can be written as
\begin{equation}
\frac{d\mathbb{\zeta}_0}{dt}=\sum_{i=1}^2 b_{i}(t) X_i^H(\mathbb{\zeta}_0),\qquad\qquad \forall\mathbb{\zeta}_0\in  \mathbb{\rom6}_0,
\end{equation}
for arbitrary t-dependent functions $ b_{i}(t).$

Since the associated  t-dependent vector field
$X^{\mathbb{\rom6}_0}=\sum_{i=1}^2 b_{i}(t) X_i^H$ takes values  in the Lie algebra $A_2$ 
,that  is,\;$[X_1^H, X_2^H]=X_1^H$
 , then $X^{\mathbb{\rom6}_0}$ is a Lie system.
 
 We now prove that
 $( \mathbb{\rom6}_0,  \Lambda_{\mathbb{\rom6}_0}, E_{\mathbb{\rom6}_0}, X^{\mathbb{\rom6}_0})$ 
 is a Jacobi-Lie  system. As a matter of fact, $X_1^H$ and $ X_2^H$   are Hamiltonian vector fields
relative to 
 $( \mathbb{\rom6}_0,  \Lambda_{\mathbb{\rom6}_0}, E_{\mathbb{\rom6}_0} )$
  with  Hamiltonian 
 functions     
 $f_1=1$
and
$f_2=x $ respectively. Subsequently,
  $( \mathbb{\rom6}_0,  \Lambda_{\mathbb{\rom6}_0}, E_{\mathbb{\rom6}_0},,X^{\mathbb{\rom6}_0 })$ 
 is a Jacobi-Lie  system.
 
Using the Lie bracket  induced by
 $ \Lambda_{ \mathbb{\rom6}_0}$  and $E_{ \mathbb{\rom6}_0}$ of Lie group $ \mathbb{\rom6}_0$, 
 we can write
$\lbrace f_1,f_2\rbrace_{\Lambda_{ \mathbb{\rom6}_0,  E_{ \mathbb{\rom6}_0} }}= f_1;$
therefore,
$( \mathbb{\rom6}_0,  \Lambda_{\mathbb{\rom6}_0}, E_{\mathbb{\rom6}_0},f=\sum_{i=1}^2 b_{i}(t) f_i )$
for   $X^{\mathbb{\rom6}_0}$ 
 is a Jacobi-Lie   Hamiltonian system.
\\ 
 It is easy to show  that
$h={{\rm e}^{2\,z}}-1$
is a t-independent constant of motion.
One can   check that
$$\lbrace   h\, ,f_\alpha\rbrace_{\Lambda_{\mathbb{\rom6}_0}, E_{\mathbb{\rom6}_0}}=0,\qquad \alpha=1,2.$$
Then, the function $h$ always Jacobi commutes with the whole Lie algebra $Lie(\lbrace f_t \rbrace _{ t\in  \mathbb{R}} ,\lbrace.,.\rbrace_{ \Lambda_{\mathbb{\rom6}_0}, E_{\mathbb{\rom6}_0}}),$ as expected.

 By applying
proposition \ref{pro},
 $X_h={\Lambda}^{\#}(dh)+h E$   must be a Lie symmetry for this system. A short calculation  shows that
$X_h=(1-{{\rm e}^{2\,z}})\partial y .$
It is easy to check that $X_h$ commutes with $X_1^H, X_2^H,$ and subsequently with
every $X_{f_{t}}, $ with $t \in \mathbb{R},$ i.e. $X_h$ is a Lie symmetry for $X^{\mathbb{\rom6}_0}$.

\section{Concluding remarks}

Using the realizations \textcolor{red}{ [\ref{realiz}]}  of the complete list of Jacobi structures on real two and three- dimensional Jacobi-Lie
groups \textcolor{red}{ [\ref{ref2}]}, we attained Hamiltonian vector fields and we   achieved Jacobi-Lie Hamiltonian systems on real
low dimensional Jacobi-Lie groups.
Then we presented the Lie symmetries of these Jacobi-Lie Hamiltonian systems.


\newpage
{\small {\bf Table 1}}: {\small
 vector field $E$ and 2-vector $\Lambda$  related to real two-dimensional bi-r-matrix Jacobi-Lie bialgebras.}\\
    \begin{tabular}{|l| p{10.76cm}|   }
    \hline\hline
  $( (\mathfrak{g},\phi_0),(\mathfrak{g^*}, X_0)) $  & vector field $E$ and 2-vector $\Lambda$  \\
 &\\
\hline
&\\

 $(( A_1,\tilde{X}^1),(A_1, X_2)) $  & $E=-\partial _y$\\
 &\\
 &$\Lambda=(1-e^{-x})\partial _x\wedge \partial _y$\\
  &\\
\hline
&\\
$(( A_2,b\tilde{X}^2),(A_2.i, -bX_1)) $  &$E=b\partial _x$\\
&\\
&$\Lambda=(1-e^{-(b+1)y})\partial _x\wedge \partial _y$\\

 &\\
 \hline
&\\
 $(( A_1,0),(A_2, -X_2)) $  & $E=\partial _y$\\
 &\\
 &$\Lambda=0$\\
\hline

        \end{tabular}
  
  \vspace{2cm}

{\small {\bf Table 2}}: {\small
vector field $E$ and 2-vector $\Lambda$  related to real three-dimensional bi-r-matrix Jacobi-Lie bialgebras.} \\
    \begin{tabular}{|l| p{10cm}|  }
    \hline\hline
  $( (\mathfrak{g},\phi_0),(\mathfrak{g^*}, X_0)) $  & vector field $E$ and 2-vector $\Lambda$  \\
 &\\
\hline
&\\
 $(( \rom{1},-\tilde{X}^2+\tilde{X}^3),(\rom{3}, -2X_1)) $  & $E=2\partial _x$\\
 &\\
 &$\Lambda=(-1+e^{y-z})\partial _x\wedge \partial _y+(1-e^{y-z})\partial _x\wedge \partial _z$\\
 
  &\\
\hline
&\\
$(( \rom{2},0),(\rom{1}, X_1)) $  &$E=-\partial _x$\\
&\\
&$\Lambda=-y\partial _x\wedge \partial _y-z\partial _x\wedge \partial _z$\\

  &\\
\hline
&\\
$(( \rom{2},0),(\rom{5}, bX_1)) $  &$E=-b\partial _x$\\
&\\
&$\Lambda=-(1+b)y\partial _x\wedge \partial _y-(1+b)z\partial _x\wedge \partial _z$\\

 &\\
\hline
&\\
$(( \rom{3},b\tilde{X}^1),(\rom{3}.i,- X_2+X_3)) $  &$E=\partial _y-\partial _z$\\
&\\
&$\begin{array}{l}\Lambda=-\frac{1}{b}(1-e^{-bx})\partial _x\wedge \partial _y+\frac{1}{b}(1-e^{-bx})\partial _x\wedge \partial _z\\\qquad-\frac{2}{b} (y+z)e^{-bx}\partial _y\wedge \partial _z\end{array}$\\

 &\\
  \hline

        \end{tabular}

   \newpage  
      {\small {\bf Table 2}}: {\small
(Continued.)}\\
    \begin{tabular}{|l| l|  }
    \hline\hline
  $( (\mathfrak{g},\phi_0),(\mathfrak{g^*}, X_0)) $  & vector field $E$ and 2-vector $\Lambda$  \\

   &\\
  \hline
&\\
$(( \rom{3},-b\tilde{X}^2+b\tilde{X}^3),(\rom{3}.iv,b X_1)) $  &$E=-b\partial _x$\\
&\\
&$\begin{array}{l}\Lambda=\frac{1}{2}(1-e^{b(y-z)})\partial _x\wedge \partial _y-\frac{1}{2}(1-e^{b(y-z)})\partial _x\wedge \partial _z\\ \qquad+(y+z)e^{b(y-z)}\partial _y\wedge \partial _z\end{array}$\\

 &\\
 \hline
&\\
$(( \rom{3},0),(\rom{3}.v,\frac{1}{2}X_2-\frac{1}{2}X_3)) $  &$E=-\frac{1}{2}\partial _y+\frac{1}{2}\partial _z$\\
&\\
&$\begin{array}{l}\Lambda=\dfrac{e^{2x}-1}{4}\partial _x\wedge \partial _y+\dfrac{e^{2x}-1}{4}\partial _x\wedge \partial _z
\\ \qquad-\frac{1}{2}(y+z+1-e^{2x})\partial _y\wedge \partial _z\end{array}$\\

 &\\
\hline
&\\
$(( \rom{3},0),(\rom{4}.iv,X_2-X_3)) $  &$E=-\partial _y+\partial _z$\\
&\\
&$\Lambda=\frac{1}{2}(e^{2x}-1)\partial _y\wedge \partial _z$\\

 &\\
\hline
&\\
$(( \rom{3},-2\tilde{X}^1),(\rom{5}.i,-X_2-X_3)) $  &$E=e^{2x}\partial _y+e^{2x}\partial _z$\\
&\\
&$\Lambda=0$\\
 &\\
\hline

\hline
&\\
$(( \rom{3},0),(\rom{6}_0.iv,X_2-X_3)) $  &$E=-\partial _y+\partial _z$\\
&\\
&$\Lambda=-2(y+z)\partial _y\wedge \partial _z$\\

  &\\
\hline
&\\
$(( \rom{3},0),(\rom{6}_a.vii,-X_2+X_3)) $  &$E=\partial _y-\partial _z$\\
&\\
&$\Lambda=\dfrac{-2}{a-1}(y+z)\partial _y\wedge \partial _z$\\ 
  
 &\\
\hline
&\\
$(( \rom{3},0),(\rom{6}_a.viii,-X_2+X_3))  $ &$E=\partial _y-\partial _z$\\
&\\
&$\Lambda=\dfrac{2}{a+1}(y+z)\partial _y\wedge \partial _z$
\\
 &\\
\hline
&\\
$(( \rom{4},-\tilde{X}^1),(\rom{3}.vi,-X_2-X_3)) $  &$E=e^{x}\partial _y+e^{x}(1-x)\partial _z$\\
&\\
&$\Lambda=-xe^{x}\partial _x\wedge \partial _z+e^{x}(z-y-1+e^{x})\partial _y\wedge \partial _z$
\\

 &\\

\hline

        \end{tabular}
     \newpage
   {\small {\bf Table 2}}: {\small
(Continued.)}\\
    \begin{tabular}{|l| p{7cm}|  }
    \hline\hline
  $( (\mathfrak{g},\phi_0),(\mathfrak{g^*}, X_0)) $  & vector field $E$ and 2-vector $\Lambda$  \\
  &\\
  \hline
&\\
$(( \rom{4},-\tilde{X}^1),(\rom{4}.i,-bX_3)) $  &$E=be^{x}\partial _z$\\
&\\
&$\Lambda=e^{x}(e^{x}-1)\partial _y\wedge \partial _z$
\\
 &\\
  \hline
&\\
$(( \rom{4},-\tilde{X}^1),(\rom{4}.ii,-bX_3)) $  &$E=be^{x}\partial _z$\\
&\\
&$\Lambda=-e^{x}(e^{x}-1)\partial _y\wedge \partial _z$
\\
 &\\
\hline
&\\
$(( \rom{4},-\tilde{X}^1),(\rom{6}_0.i,-X_3)) $  &$E=e^{x}\partial _z$\\
&\\
&$\Lambda=-2ye^{x}\partial _y\wedge \partial _z$
\\
 &\\
\hline
&\\
$(( \rom{4},-\tilde{X}^1),(\rom{6}_a.i,-X_3)) $  &$E=e^{x}\partial _z$\\
&\\
&$\Lambda=\dfrac{2}{a-1}ye^{x}\partial _y\wedge \partial _z$
\\
 &\\

  \hline
&\\
$(( \rom{4},-\tilde{X}^1),(\rom{6}_a.ii,-X_3)) $  &$E=e^{x}\partial _z$\\
&\\
&$\Lambda=-\dfrac{2}{a+1}ye^{x}\partial _y\wedge \partial _z$
\\
 &\\
 \hline
&\\
 $(( \rom{5},-2\tilde{X}^1),(\rom{5}.i, -2X_2-2X_3)) $  & $E=2e^{x}\partial _y+2e^{x}\partial _z$\\
 &\\
 &$ \begin{array}{l}\Lambda=e^{x}(1-e^{x})\partial _x\wedge \partial _y+e^{x}(1-e^{x})\partial _x\wedge \partial _z\\ \qquad+e^{2x}(z-y)\partial _y\wedge \partial _z\end{array}$\\
 
  &\\
\hline
&\\
 $(( \rom{5},-\frac{2a}{a-1}\tilde{X}^1),(\rom{6}_a.i, -\frac{2a}{a-1}X_3)) $  & $E=\frac{2a}{a-1}e^{x}\partial _z$\\
 &\\
 &$\Lambda=(e^{x}-e^{\frac{2a}{a-1}x})\partial _x\wedge \partial _z-ye^{\frac{2a}{a-1}x}\partial _y\wedge \partial _z$\\
 
  &\\
\hline
&\\
 $(( \rom{5},-\frac{2a}{a+1}\tilde{X}^1),(\rom{6}_a.ii, -\frac{2a}{a+1}X_3))$  & $E=\frac{2a}{a+1}e^{x}\partial _z$\\
 &\\
 &$\Lambda=(e^{x}-e^{\frac{2a}{a+1}x})\partial _x\wedge \partial _z-ye^{\frac{2a}{a+1}x}\partial _y\wedge \partial _z$\\
 
 &\\
  
\hline

        \end{tabular}
        
     \newpage
     {\small {\bf Table 2}}: {\small
(Continued.)}\\
    \begin{tabular}{|l| l|   }
\hline\hline
  $( (\mathfrak{g},\phi_0),(\mathfrak{g^*}, X_0)) $  & vector field $E$ and 2-vector $\Lambda$  \\
  &\\
\hline

&\\
$(( \rom{6}_0,\tilde{X}^3),(\rom{3}.vii,- X_1-X_2)) $  &$E=\partial _x+\partial _y$\\
&\\
&$ \begin{array}{l}\Lambda=(y-x)\partial _x\wedge \partial _y+(1-e^{-2z})\partial _x\wedge \partial _z\\ \qquad+(1-e^{-2z})\partial _y\wedge \partial _z\end{array}$
\\
 &\\
\hline
&\\
$(( \rom{6}_0,\tilde{X}^3),(\rom{3}.ix,- X_1))  $  &$E=\partial _x$\\
&\\
&$ \begin{array}{l}   \Lambda=(1+y-e^{-z})\partial _x\wedge \partial _y+e^{-z}\sinh(z)\partial _x\wedge \partial _z\\ \qquad+(1-e^{-z}\cosh(z))\partial _y\wedge \partial _z\end{array} $
\\
 &\\
\hline
&\\
$(( \rom{6}_0,\tilde{X}^3),(\rom{6}_0.ii,- X_1+X_2)) $  &$E=\partial _x-\partial _y$\\
&\\
&$\Lambda=2(x+y)\partial _x\wedge \partial_y$
\\
 &\\
\hline
&\\
$(( \rom{6}_0,-2\tilde{X}^3),(\rom{6}_0.ii,2X_1-2X_2)) $  &$E=-2\partial _x+2\partial _y
$\\
&\\
&$\begin{array}{l}\Lambda=-(x+y)\partial _x\wedge \partial _y+(1-e^{3z})\partial _x\wedge \partial _z\\ \qquad-(1-e^{3z})\partial _y\wedge \partial _z\end{array}$
\\
 &\\
\hline
&\\
$(( \rom{6}_0,\tilde{X}^3),(\rom{6}_a.iii,-X_1+X_2)) $  &$E=\partial _x-\partial _y$\\
&\\
&$\Lambda=-\frac{2}{a-1}(x+y)\partial _x\wedge \partial _y$
\\
 &\\
\hline
&\\
$(( \rom{6}_0,\tilde{X}^3),(\rom{6}_a.iv,-X_1+X_2)) $  &$E=\partial _x-\partial _y$\\
&\\
&$\Lambda=\frac{2}{a+1}(x+y)\partial _x\wedge \partial _y$
\\
 &\\
\hline
&\\
$(( \rom{6}_0,\frac{2}{a-1}\tilde{X}^3),(\rom{6}_a.iii,-\frac{2}{a-1}(X_1-X_2))) $  &$E=\frac{2}{a-1}\partial _x-\frac{2}{a-1}\partial _y$\\
&\\
&$\begin{array}{l}\Lambda=-(x+y)\partial _x\wedge \partial _y+(1-e^{\frac{a-3}{a-1}z})\partial _x\wedge \partial _z\\ \qquad-(1-e^{\frac{a-3}{a-1}z})\partial _y\wedge \partial _z
\end{array}$
\\
 &\\
\hline

&\\
$(( \rom{6}_0,-\frac{2}{a+1}\tilde{X}^3),(\rom{6}_a.iv,\frac{2}{a+1}(X_1-X_2))) $  &$E=-\frac{2}{a+1}\partial _x+\frac{2}{a+1}\partial _y$\\
&\\
&$\begin{array}{l}  \Lambda=-(x+y)\partial _x\wedge \partial _y+(1-e^{\frac{a+3}{a+1}z})\partial _x\wedge \partial _z\\ \qquad-(1-e^{\frac{a+3}{a+1}z})\partial _y\wedge \partial _z
\end{array}$\\
\hline
        \end{tabular}
        \newpage
     {\small {\bf Table 2}}: {\small
(Continued.)}\\
    \begin{tabular}{|l| l|   }
\hline\hline
  $( (\mathfrak{g},\phi_0),(\mathfrak{g^*}, X_0)) $  & vector field $E$ and 2-vector $\Lambda$  \\
  &\\
  \hline
  &\\
 $(( \rom{6}_a,-(a-1)\tilde{X}^1),(\rom{3}.ii, -\frac{a-1}{a+1}(X_2+X_3))) $  & $E=\frac{a-1}{a+1}e^{(a+1)x}\partial _y+\frac{a-1}{a+1}e^{(a+1)x}\partial _z$\\
 &\\
&$\begin{array}{l}   \Lambda=\frac{1}{a+1}(e^{(a+1)x}-e^{(a-1)x})\partial _x\wedge \partial _y  \\ \qquad+\frac{1}{a+1}(e^{(a+1)x}-e^{(a-1)x})\partial _x\wedge \partial _z\\
\qquad +\frac{1}{a+1}e^{(a-1)x}(1-a)(y-z)\partial _y\wedge \partial _z\end{array}$\\
  &\\
\hline
  &\\
 $(( \rom{6}_a, -(a+1)\tilde{X}^1),(\rom{3}.v, \frac{1}{a-1}(X_2-aX_3))) $  &
 $\begin{array}
{l}E=\frac{1}{a-1}e^{ax}(-\cosh(x)+a\sinh(x))\partial _y\\ \qquad+\frac{1}{a-1}e^{ax}(-\sinh(x)+a\cosh(x))\partial _z
 \end{array}$
  \\
 &
$\begin{array}{l}\Lambda=\frac{1}{a-1}e^{ax}\sinh(x)\partial _x\wedge \partial _y\\ \qquad+\frac{1}{a-1}(e^{ax}\cosh(x)-e^{(a+1)x})\partial _x\wedge \partial _z\\ \qquad+\frac{1}{a-1}(e^{2ax}-e^{(a+1)x}(1+ay+z))\partial _y\wedge \partial _z\end{array}$
 \\
 &\\
   &\\
\hline
&\\
 $(( \rom{6}_a,-(a-1)\tilde{X}^1),(\rom{3}.v, \frac{1}{a+1}(X_2-aX_3)))$  & 
  $\begin{array}{l}E=\frac{1}{a+1}e^{ax}(-\cosh(x)+a\sinh(x))\partial _y\\\qquad+\frac{1}{a+1}e^{ax}(-\sinh(x)+a\cosh(x))\partial _z\end{array}$\\
 &\\
 & $\begin{array}{l}\Lambda=\frac{1}{a+1}e^{ax}\sinh(x)\partial _x\wedge \partial _y\\\qquad+\frac{1}{a+1}(e^{ax}\cosh(x)-e^{(a-1)x})\partial _x\wedge \partial _z\\\qquad+\frac{1}{a+1}(e^{2ax}-e^{(a-1)x}(1+ay+z))\partial _y\wedge \partial _z\end{array}$\\
 &\\
 &\\
\hline
&\\
&\\
$(( \rom{6}_a,-(a+1)\tilde{X}^1),(\rom{3}.x,-\frac{a+1}{a-1}( X_2-X_3))) $  &$E=\frac{a+1}{a-1}e^{(a-1)x}\partial _y-\frac{a+1}{a-1}e^{(a-1)x}\partial _z$\\
&\\
&$\begin{array}{l} \Lambda=\frac{1}{a-1}(e^{(a-1)x}-e^{(a+1)x})\partial _x\wedge \partial _y\\
\qquad +\frac{1}{a-1}(e^{(a+1)x}-e^{(a-1)x})\partial _x\wedge \partial _z\\
 \qquad+\frac{1}{a-1}e^{(a+1)x}(1+a)(y+z)\partial _y\wedge \partial _z \end{array}$
\\
&\\
\hline
&\\
$(( \rom{6}_a,-(a+1)\tilde{X}^1),(\rom{6}_b.v,-X_2-X_3))  $  &$E=e^{(a+1)x}\partial _y+e^{(a+1)x}\partial _z$\\
&\\
&$\Lambda=\frac{2}{(b-1)}e^{(a+1)x}(y-z)\partial _y\wedge \partial _z$
\\
&\\
\hline

  &\\
$(( \rom{6}_a,-(a+1)\tilde{X}^1),(\rom{6}_b.vi,- X_2+X_3)) $  &$E=e^{(a+1)x}\partial _y+e^{(a+1)x}\partial _z$\\
&\\
&$\Lambda=-\frac{2}{(b+1)}e^{(a+1)x}(y-z)\partial _y\wedge \partial _z$
\\
&\\
\hline

        \end{tabular}
  
  \newpage
 {\small {\bf Table 2}}: {\small
(Continued.)}\\
    \begin{tabular}{|l| l|  }
    \hline\hline
  $( (\mathfrak{g},\phi_0),(\mathfrak{g^*}, X_0)) $  & vector field $E$ and 2-vector $\Lambda$  \\
  &\\
\hline 
  &\\
  $(( \rom{6}_a,-(a-1)\tilde{X}^1),(\rom{6}_b.vii,-X_2+X_3)) $  &$E=e^{(a-1)x}\partial _y-e^{(a-1)x}\partial _z$\\
&\\
&$\Lambda=-\frac{2}{(b-1)}e^{(a-1)x}(y+z)\partial _y\wedge \partial _z$
\\
&\\
\hline
  $(( \rom{6}_a,-(a-1)\tilde{X}^1),(\rom{6}_b.viii,-X_2+X_3)) $  &$E=e^{(a-1)x}\partial _y-e^{(a-1)x}\partial _z$\\
&\\
&$\Lambda=\frac{2}{(b+1)}e^{(a-1)x}(y+z)\partial _y\wedge \partial _z$
\\
&\\
\hline
&\\
$(( \rom{6}_a,-\frac{2(ab+1)}{b-1}\tilde{X}^1),(\rom{6}_b.v,-\frac{2(ab+1)}
{(a+1)(b-1)}(X_2+X_3)))$ &
$\begin{array}{l}E=\frac{2(ab+1)}{(a+1)(b-1)}e^{(a+1)x}\partial _y \\ \qquad+\frac{2(ab+1)}{(a+1)(b-1)}e^{(a+1)x}\partial _z\end{array} $\\
&\\
&$\begin{array}{l}\Lambda=\frac{1}{a+1}(e^{(a+1)x}-e^{\frac{2(ab+1)}{b-1}x})\partial _x\wedge \partial _y\\
\qquad+\frac{1}{a+1}(e^{(a+1)x}-e^{\frac{2(ab+1)}{b-1}x})\partial _x\wedge \partial _z\\
\qquad+\frac{1}{a+1}e^{\frac{2(ab+1)}{b-1}x}(1-a)(y-z)\partial _y\wedge \partial _z\end{array}$
\\
&\\
\hline
  
  &\\
  $(( \rom{6}_a, -\frac{2(ab-1)}{b+1}\tilde{X}^1),(\rom{6}_b.vi,-\frac{2(ab-1)}{(a+1)(b+1)}(X_2+X_3))) $  &
  $\begin{array}{l} E=\frac{2(ab-1)}{(a+1)(b+1)}e^{(a+1)x}\partial _y\\ \qquad+\frac{2(ab-1)}{(a+1)(b+1)}e^{(a+1)x}\partial _z\end{array} $\\
  &\\
&$\begin{array}{l} \Lambda=\frac{1}{a+1}(e^{(a+1)x}-e^{\frac{2(ab-1)}{b+1}x})\partial _x\wedge \partial _y\\
\qquad+\frac{1}{a+1}(e^{(a+1)x}-e^{\frac{2(ab-1)}{b+1}x})\partial _x\wedge \partial _z\\
\qquad+\frac{1}{a+1}e^{\frac{2(ab-1)}{b+1}x}(1-a)(y-z)\partial _y\wedge \partial _z\end{array} $
\\
&\\  
\hline
&\\
$(( \rom{6}_a, -\frac{2(ab-1)}{b-1}\tilde{X}^1),(\rom{6}_b.vii,-\frac{2(ab-1)}{(a1)(b-1)}(X_2-X_3))) $  &$\begin{array}{l}E=\frac{2(ab-1)}{(a-1)(b-1)}e^{(a-1)x}\partial _y\\ \qquad -\frac{2(ab-1)}{(a-1)(b-1)}e^{(a-1)x}\partial _z\end{array}$\\
&\\
&$\begin{array}{l}\Lambda=\frac{1}{a-1}(e^{(a-1)x}-e^{\frac{2(ab-1)}{b-1}x})\partial _x\wedge \partial _y\\
\qquad +\frac{1}{a-1}(-e^{(a-1)x}+e^{\frac{2(ab-1)}{b-1}x})\partial _x\wedge \partial _z\\
\qquad +\frac{1}{a-1}e^{\frac{2(ab-1)}{b-1}x}(1+a)(y+z)\partial _y\wedge \partial _z\end{array}$\\
&\\
\hline
&\\
$(( \rom{6}_a, -\frac{2(ab+1)}{b+1}\tilde{X}^1),(\rom{6}_b.viii,-\frac{2(ab+1)}{(a-1)(b+1)}(X_2-X_3))) $  &$\begin{array}{l}E=\frac{2(ab+1)}{(a-1)(b+1)}e^{(a-1)x}\partial _y\\ \qquad-\frac{2(ab+1)}{(a-1)(b+1)}e^{(a-1)x}\partial _z\end{array}$\\
&\\
&$\begin{array}{l}\Lambda=\frac{1}{a-1}(e^{(a-1)x}-e^{\frac{2(ab+1)}{b+1}x})\partial _x\wedge \partial _y\\
\qquad 
+\frac{1}{a-1}(-e^{(a-1)x}+e^{\frac{2(ab+1)}{b+1}x})\partial _x\wedge \partial _z\\
\qquad +\frac{1}{a-1}e^{\frac{2(ab+1)}{b+1}x}(1+a)(y+z)\partial _y\wedge \partial _z\end{array}
$
\\
&\\
\hline

  \hline
        \end{tabular}
\newpage
{\small {\bf Table 3}}: {\small
 vector field $E$ and 2-vector $\Lambda$  related to real three dimensional coboundary Jacobi-Lie bialgebras.}\\
    \begin{tabular}{|l| p{9.5cm}|  }
    \hline\hline
  $( (\mathfrak{g},\phi_0),(\mathfrak{g^*}, X_0)) $  & vector field $E$ and 2-vector $\Lambda$  \\
  &\\
\hline
&\\
 $(( \rom{1},0),(\rom{5}, -X_1)) $  & $E=\partial _x$\\
 &\\
 &$\Lambda=0$\\
  &\\
\hline
&\\
$(( \rom{2},-\tilde{X}^2+\tilde{X}^3),(\rom{3}, -2X_1)) $  &$E=2\partial _x$\\
&\\
&$\Lambda=(-1+e^{y-z})\partial _x\wedge \partial _y+(1-e^{y-z})\partial _x\wedge \partial _z$
\\
  &\\
\hline
&\\
$(( \rom{3},-\tilde{X}^2+\tilde{X}^3),(\rom{3}.iii, X_2+X_3)) $  &$E=-e^{2x}\partial _y-e^{2x}\partial _z$\\
&\\
&$\Lambda=(-e^{2x}+e^{y-z})\partial _y\wedge \partial _z$
\\
 &\\
\hline
&\\
$(( \rom{3},0),(\rom{3}.x,- X_2+X_3)) $  &$E=\partial _y-\partial _z$\\
&\\
&$\Lambda=(y+z)\partial _y\wedge \partial _z$
\\
 &\\
\hline
&\\
$(( \rom{4},-\tilde{X}^1),(\rom{3}.v,-X_3)) $  &$E=e^{x}\partial _z$\\
&\\
&$\Lambda=-(y+1-e^{x})e^{x}\partial _y\wedge \partial _z$
\\
 &\\
\hline
&\\
$(( \rom{4},-2\tilde{X}^1),(\rom{5}.ii,-2X_3)) $  &$E=2e^{x}\partial _z$\\
&\\
&$\Lambda=(e^{x}-e^{2x})\partial _x\wedge \partial _z-ye^{2x}\partial _y\wedge \partial _z$
\\
 &\\
\hline
&\\
$(( \rom{4},-\frac{2a}{a-1}\tilde{X}^1),(\rom{6}_a.i,-\frac{2a}{a-1}X_3)) $  &$E=\frac{2a}{a-1}e^{x}\partial _z$\\
&\\
&$\Lambda=(e^{x}-e^{\frac{2a}{a-1}x})\partial _x\wedge \partial _z-ye^{\frac{2a}{a-1}x}\partial _y\wedge \partial _z$
\\
 &\\
\hline
&\\
$(( \rom{4},-\frac{2a}{a+1}\tilde{X}^1),(\rom{6}_a.ii,-\frac{2a}{a+1}X_3)) $  &$E=\frac{2a}{a+1}e^{x}\partial _z$\\
&\\
&$\Lambda=(e^{x}-e^{\frac{2a}{a+1}x})\partial _x\wedge \partial _z-ye^{\frac{2a}{a+1}x}\partial _y\wedge \partial _z$
\\

 \hline
        \end{tabular}
        \newpage
       {\small {\bf Table 3}}: {\small
(Continued.)}\\
    \begin{tabular}{|l| p{7cm}|   }
    \hline\hline
  $( (\mathfrak{g},\phi_0),(\mathfrak{g^*}, X_0)) $  & vector field $E$ and 2-vector $\Lambda$  \\
 &\\
 \hline
 &\\
 $(( \rom{5},-\tilde{X}^1),(\rom{6}_0.i,-X_3)) $  &$E=e^{x}\partial _z$\\
&\\
&$\Lambda=-2ye^{x}\partial _y\wedge \partial _z$
\\
 &\\
 \hline
 &\\
$(( \rom{5},-\tilde{X}^1),(\rom{6}_a.i,-X_3)) $  &$E=e^{x}\partial _z$\\
&\\
&$\Lambda=\frac{2}{a-1}ye^{x}\partial _y\wedge \partial _z$
\\
 &\\
\hline
&\\
$(( \rom{5},-\tilde{X}^1),(\rom{6}_a.ii,-X_3)) $  &$E=e^{x}\partial _z$\\
&\\
&$\Lambda=-\frac{2}{a+1}ye^{x}\partial _y\wedge \partial _z$
\\
 &\\
\hline
&\\
$(( \rom{6}_0,\tilde{X}^3),(\rom{3}.viii,-X_1+X_2)) $  &$E=\partial _x-\partial _y$\\
&\\
&$\Lambda=(x+y)\partial _x\wedge \partial _y$
\\
 &\\
\hline
&\\
$(( \rom{6}_a,-(a+1)\tilde{X}^1),(\rom{3}.ii,-X_2-X_3)) $  &$E=e^{(a+1)x}\partial _y+e^{(a+1)x}\partial _z$\\
&\\
&$\Lambda=-(y-z)e^{(a+1)x}\partial _y\wedge \partial _z$
\\
 &\\
\hline
&\\
$(( \rom{6}_a,-(a-1)\tilde{X}^1),(\rom{3}.x,-X_2+X_3)) $  &$E=e^{(a-1)x}\partial _y+e^{(a-1)x}\partial _z$\\
&\\
&$\Lambda=(y+z)e^{(a-1)x}\partial _y\wedge \partial _z$
\\
 &\\
\hline
        \end{tabular}

\end{document}